\documentclass[10pt, conference, letterpaper]{IEEEtran}
\IEEEoverridecommandlockouts
\usepackage{latexsym}
\usepackage{amsmath}
\usepackage{amsthm}
\usepackage{amssymb}
\usepackage{verbatim}
\usepackage{setspace}
\usepackage{algorithm}
\usepackage[noend]{algpseudocode}
\usepackage{hyperref, url}
\usepackage{ifthen}
\usepackage{graphicx, color, xcolor}
\usepackage{subcaption}
\usepackage{relsize}

\newtheorem{theorem}{Theorem}
\newtheorem{lemma}[theorem]{Lemma}

\newtheorem{corollary}[theorem]{Corollary}
\newtheorem{definition}{Definition}
\newtheorem{fact}[definition]{Fact}

\newcommand{\defeq}{\stackrel{\mathsmaller{\mathsf{def}}}{=}}

\algblock[phase]{Phase}{EndPhase}
\algblock[initialize]{Init}{EndInit}
\algblock[receive]{OnReceive}{EndReceive}
\algblock[foreach]{foreach}{endForeach}

\makeatletter
\ifthenelse{\equal{\ALG@noend}{t}}%
  {\algtext*{EndInit}}
  {}%
\makeatother

\makeatletter
\ifthenelse{\equal{\ALG@noend}{t}}%
  {\algtext*{EndReceive}}
  {}%
\makeatother

\makeatletter
\ifthenelse{\equal{\ALG@noend}{t}}%
  {\algtext*{endForeach}}
  {}%
\makeatother

\long\gdef\boxit#1{\vspace{5mm}\begingroup\vbox{\hrule\hbox{\vrule\kern3pt
\vbox{\kern3pt#1\kern3pt}\kern3pt\vrule}\hrule}\endgroup}

\begin{document}

\title{Fast and Efficient Distributed Computation of Hamiltonian Cycles in Random Graphs\thanks{Supported, in part, by NSF grants CCF-1527867, CCF-1540512,  IIS-1633720,  CCF-1717075, and
BSF award 2016419.}}

\author{
    \IEEEauthorblockN{Soumyottam Chatterjee, Reza Fathi, Gopal Pandurangan, Nguyen Dinh Pham}

    \IEEEauthorblockA{Department of Computer Science, University of Houston, Houston,
Texas, 77204, USA \! \\
{\tt soumyottam@acm.org}, {\tt rfathi@cs.uh.edu}, {\tt gopalpandurangan@gmail.com}, {\tt aphamdn@gmail.com} }}

\maketitle

\begin{abstract}
  We present  fast  and efficient randomized distributed algorithms to find Hamiltonian cycles in random  graphs. In particular, we present a randomized distributed algorithm for the $G(n,p)$ random graph model,
  with number of nodes $n$ and $p=\frac{c\ln n}{n^{\delta}}$ (for any constant $0 < \delta \leq 1$ and for a suitably large constant $c > 0$),  that finds a Hamiltonian cycle with high
  probability in $\tilde{O}(n^{\delta})$ rounds.\footnote{The notation $\tilde{O}$ hides a $\text{polylog}(n)$ factor.}  Our algorithm works in the (synchronous) CONGEST model (i.e., only $O(\log n)$-sized messages are communicated per edge per round) and its computational cost per node is  sublinear  (in $n$) per round and is fully-distributed (each node uses only $o(n)$ memory and all nodes' computations are essentially balanced).
    Our algorithm improves over the previous best known result in terms of both the running time as well as the edge sparsity of the graphs where it can succeed; in particular, the denser the random graph, the smaller is the running time.
\end{abstract}

\section{Introduction}
%Graphs are natural data structures to store and manipulate when there are relations between
%entities.

Finding Hamiltonian cycles (or paths) in graphs (networks) is one of the  fundamental graph problems.
Hamiltonian cycle (HC) is a cycle in the graph that passes through each node exactly once.
The decision problem is NP-complete \cite{garey} (in fact, it is one of Karp's six basic NP-complete problems) and hence unlikely to have a polynomial time algorithm in the sequential setting.
In this paper,
we focus on the distributed computation of
  Hamiltonian cycles (or paths) in a (undirected) graph. In particular, our goal is to find a {\em fast}, {\em efficient}, and {\em fully}  distributed  algorithm for the Hamiltonian cycle problem. By ``fast'', we mean running in a small number of {\em rounds} (ideally, sublinear in $n$, where $n$ is the number of nodes in the network). By ``efficient'', we mean that  only small-sized messages (say, at most $O(\log n)$-sized messages) are exchanged per edge per round, and the per-round computation per node  should  also be small, i.e.,  sublinear in $n$. The latter means that the local (i.e., ``within node") computation is also efficient. By ``fully-distributed'', we (informally) mean that no one node (or a small set of nodes) does all the non-trivial (local) computation and all the local computations are (more or less) balanced (formally we enforce this by  assuming
  that each node's  memory is limited to $o(n)$).

Since the HC problem is NP-complete, there is not much hope of achieving a fast and efficient distributed  algorithm (even if we allow polynomial time local computation per round and even without caring whether it is  fully-distributed or not)
in arbitrary graphs, even if we allow polynomial number of rounds (since  the  total local computation time over all nodes is at most polynomial).
However, the problem is reasonable and, yet challenging, when we consider random graphs, where efficient sequential algorithms (nearly linear time) for computing Hamiltonian cycles are known.

Despite the importance of the Hamiltonian cycle problem, there has been only some previous work  in the distributed setting. The work of Das Sarma et al \cite{stoc11} (see also \cite{podc14}) showed an important lower bound for the HC problem for general graphs in the CONGEST model of distributed computing \cite{peleg-book} (described in detail in Section \ref{sec:model}), a standard model where there is a bandwidth restriction on the edges (typically,
only $O(\log n)$-sized messages are allowed per edge per round, where $n$ is the graph/network size).
They showed that any deterministic algorithm (this was extended to hold even for {\em randomized} algorithms
in \cite{podc14}) needs at least $\tilde{\Omega}(D+ \sqrt{n})$ rounds, where $D$ is the graph diameter\footnote{The notation $\tilde{\Omega}$ hides a $1/\text{polylog} n$ factor.}. Note that this lower bound holds even if every node's local computation is free (i.e., there is no restriction on the within node computation cost in a round --- this is the usual assumption in the CONGEST model \cite{peleg-book}). It is important to note that this lower bound is for general graphs; more precisely, it holds for a family of graphs constructed in a special way.

Somewhat surprisingly, no non-trivial upper bounds are known for the distributed HC problem in the CONGEST model. A trivial upper bound in the CONGEST model is $O(m)$ where $m$ is the number of edges of the graph (cf. Section \ref{sec:model}). It is not known if one can get a $(O(D) + o(n))$-round algorithm or even a $(O(D)+ o(m))$-round algorithm for HC in general graphs, where $D$ is the graph diameter (note that
$D$ is a lower bound \cite{stoc11}).
In this paper, we show that we can obtain significantly faster (truly sublinear in $n$) algorithms, i.e., running in time
$O(n^{\delta})$ rounds (where $0 < \delta <  1$) in  random graphs.

%Then questions such as calculating Shortest paths, MST, or Hamiltonian cycle (HC) emerge.

%Building token ring in wireless networks is an application of Hamiltonian cycles for instance.
%Therefore much research has been done to give an efficient algorithm to find an HC in a graph (if one exists).

%Since graphs differ for each problem, one way to study and analyze graph algorithms that works in
%general for most of the problems is using models that present some characteristics of those real
%world graphs.

We focus on the $G(n,p)$ random graph model\cite{erdos1960evolution}, a popular and well-studied model of random graphs with a long history in the study of graph algorithms (see e.g., \cite{Bollobas-book,frieze-survey} and the references therein).
Random graphs such as $G(n,p)$ and its variants and generalizations (e.g., the Chung-Lu model \cite{chung-lu}) have been used extensively to model and analyze real-world networks.
%(It might also be possible to extend our distributed algorithms to work on other random graph models as well -- cf. Section \ref{sec:conc}).
%is a suitable choice for this
%purpose because it preserves stochastic properties of nodes' connectivity in real worlds.
In the $G(n,p)$ random graph model, there are $n$ nodes and the probability that an edge exists between any two nodes is $p$ (independent of other edges).
A remarkable property of $G(n,p)$ model is
that if $p$ is above a certain threshold, then with high probability (whp)\footnote{Throughout, by ``with high probability (whp)'', we mean a probability at least $1 - 1/n^c$, for some constant $c>0$, where $n$ is the number of nodes.},
a Hamiltonian cycle (HC) exists. More precisely,
it is known that, with high probability, for $n$ sufficiently large, there exists a HC in
  $G(n,p)$ if $p  \geq  \frac{c\ln n}{n}$, for any constant $c > 1$  \cite{palmer}; in fact, not one, but it can be shown that exponential number of Hamiltonian cycles exist \cite{glebov,cooper1989number} for $p$ above this threshold \footnote{Actually, the ``real" threshold for Hamiltonian cycles is $p \geq \frac{\ln n + \ln \ln n +\omega(1)}{n}$, if one wants to show the existence of HC asymptotically almost surely \cite{Bollobas-book}. We use a slight larger threshold, since we want algorithms that succeed to find a HC whp.}
It is worth noting that the above threshold  for $p$ is (essentially) the same as the threshold  for connectivity of a $G(n,p)$ random graph.

%Frieze A. et al
%in \cite{frieze2016perfect} showed the existence probability for even sparser graphs.
Since it is known that Hamiltonian cycles exists in $G(n,p)$ random graphs, there has been work in devising efficient  algorithms for {\em finding}  Hamiltonian cycles in these graphs.
This is a non-trivial task, even though as mentioned earlier that there are exponential number of HCs present.
%Bollobas  et al in \cite{bollobas1987algorithm} proposed a sequential algorithm to find an HC. Their algorithm
%runs in $O(n^{4+\epsilon})$ time.
Angluin and Valiant \cite{angluin1979fast}, in a seminal paper (see also \cite{mitzenmacher2017probability}),  gave a sequential algorithm  to find a HC in
a $G(n,p)$ graph that runs in $O(n(\log n)^2)$ time, when $p \geq \frac{c\ln n}{n}$, for some sufficiently large constant (say $c \geq 36$). This is essentially the best possible as far as the sequential running time is concerned as it is almost linear.
The algorithm of Angluin and Valiant is randomized. Bollobas, Fenner, and Frieze \cite{bollobas1987algorithm} give a deterministic sequential algorithm for finding Hamilton cycles in random graphs (in the related $G(n,M)$ random graph model, which is a uniform distribution over all graphs on $n$ vertices and $M$ edges), but the running time is essentially  $O(n^4)$ and succeeds with high probability (in graphs where the number of edges is above the threshold of existence of Hamiltonian cycle).
In the context of parallel algorithms,
MacKenzie  et al in \cite{mackenzie1993optimal} proposed a parallel algorithm which uses
$O(\frac{n}{\log* n})$ processes and runs in $O(\log^* n)$ time.
In the {\em distributed setting}, the only prior work we are aware of is the work of Levy et al \cite{levy2004distributed} which  gives a distributed algorithm to find a HC in
$O(n^{\frac{3}{4} + \epsilon})$ time when $p  =  \omega(\frac{\sqrt{\log{n}}}{n^{\frac{1}{4}}})$.
%Their algorithm works in three
%phases: finding an initial cycle, finding $\sqrt{n}$ disjoint paths, and finally patching paths
%into the cycle to build their final HC.

In this paper,  we propose a {\em fast, efficient, and fully decentralized} (as defined earlier) distributed algorithm
that finds a HC with high probability and runs in time significantly faster than the prior work of \cite{levy2004distributed} as well as works for all ranges of $p$; in particular, the denser the graph, the faster will be our algorithms. Our distributed algorithms that run on random graphs  are themselves randomized (i.e., they make random choices during the course of the algorithm) and hence the high probability bounds are both  with respect to the random input and the random choices of the algorithm.

We give a brief overview of our results. In Section \ref{sec:ouralg}, we give two fast (truly sublinear in $n$), efficient and fully decentralized algorithms. The first algorithm, is a bit simpler, and works for $p \geq \frac{c\ln n}{\sqrt{n}}$ and runs in $\tilde{O}(\sqrt{n})$ rounds.
The second algorithm works for $p = \frac{c\ln n}{n^{\delta}}$, for any fixed constant $\delta \in (0,1)$, and for a suitably large constant $c$, and runs in $\tilde{O}(n^{\delta})$ rounds. (Our algorithm will also work
for $\delta = 1$, with running time $\tilde{O}(n)$.) Both algorithms work in the CONGEST model and are fully distributed, i.e., no node (or a few nodes) does all the computation  (since the memory size of each node is restricted to be $o(n)$ --- cf. Section \ref{sec:model}). In contrast, in Section \ref{sec:upcastalg}, we present a (conceptually) much simpler {\em upcast} algorithm that uses a fairly generic ``centralized" approach. In this algorithm, each node samples $\Theta(\log n)$ random  edges among all its incident edges and upcasts it to a central node (which is the root of a Breadth First Tree) which locally computes a HC and then broadcasts the HC edges back to the respective nodes by downcast. Note that, in this approach, all the non-trivial (local) computation is done at a central node and hence the algorithm is not fully distributed (some node needs at least $\Omega(n)$ memory), although the algorithm works in the CONGEST model.
We show that this algorithm also runs in time $\tilde{O}(n^{\delta})$ rounds for $p = \frac{c\ln n}{n^{\delta}}$.

Before we describe our algorithms, we detail our distributed computing Model (Section \ref{sec:model})
and discuss other related work (Section \ref{sec:related}). We conclude with open questions in Section \ref{sec:conc}.

\subsection{Distributed Computing Model}
\label{sec:model}
We model the communication network as an undirected, unweighted, connected graph $G = (V, E)$, where $|V| = n$ and $|E| = m$. Every  node has limited initial knowledge. Specifically, assume that each node is associated with a distinct identity number  (e.g., its IP address).
At the beginning of the computation, each node $v$ accepts as input its own identity number and the identity numbers of its neighbors in $G$.
% The node may also accept some additional inputs as specified by the problem at hand. Here,
We also assume that the number of nodes and edges i.e., $n$ and $m$ (respectively) are given as inputs. (In any case, nodes can compute them easily through broadcast in $O(D)$, where $D$ is the network diameter.) The nodes are only allowed to communicate through the edges of the graph $G$. We assume that the communication occurs in  synchronous  {\em rounds}.
(In particular, all the nodes wake up simultaneously at the beginning of round 1, and from this point on the nodes always know the number of the current round.)
We will use only small-sized messages. In particular, in each round, each node $v$ is allowed to send a message of size $O(\log n)$ bits
through each edge $e = (v, u)$ that is adjacent to $v$.\footnote{Our algorithms can be easily generalized if $\mathbb{B}$ bits  are allowed (for any pre-specified parameter $\mathbb{B}$) to be sent through each edge in a round. Typically, as assumed here, $\mathbb{B} = O(\log n)$, which is the number of bits needed to send a node id in a $n$-node network.}  The message  will arrive to $u$ at the end of the current round.
%This is a standard model of distributed computation known as the {\em CONGEST model} \cite{peleg} and has been attracting a lot of research attention during last two decades (e.g., see \cite{peleg} and the references therein).
This is a  widely used  standard model known as the {\em CONGEST model} to study distributed algorithms (e.g., see \cite{peleg-book,PK09}) and captures the bandwidth constraints inherent in real-world computer  networks.

We  focus on minimizing the  {\em running time}, i.e., the number of {\em rounds} of distributed communication. Note that the computation that is performed by the nodes locally is ``free'', i.e., it does not affect the number of rounds; however, as mentioned earlier, we will only perform sublinear (in $n$) cost computation locally  at any node.

We note that in the CONGEST model, it is rather trivial to solve a problem in $O(m)$ rounds, where $m$ is the number of edges in the network, since the entire topology (all the edges) can be collected at one node and the problem solved locally. The goal is to design faster algorithms.
Our algorithms work in the {\em CONGEST} model of distributed computing.
  We note that our bounds are non-trivial in the CONGEST model.\footnote{In contrast, in the LOCAL model --- where there is no bandwidth constraint --- all problems can be trivially solved in $O(D)$ rounds by collecting all the topological information at one node.}

In Section \ref{sec:ouralg}, we consider fully-distributed algorithms, where there is a restriction on
the amount of memory each node can have: each node is allowed only $o(n)$ memory. This restriction, in effect, rules out ``centralized" approaches such as collecting global information at one particular node and then locally solving the problem. In our fully-distributed algorithms, each node's (local) computation is more or less balanced.
Fully-distributed algorithms are quite useful, since they can be be efficiently converted to work
in other distributed models for Big Data computing such as the $k$-machine model \cite{soda15} as well as MapReduce \cite{soda10}.

In Section \ref{sec:upcastalg}, we consider algorithms where
we don't have any restriction on the memory size at any node nor we restrict the local computation cost to be sublinear (note that this restriction turns out to be not so important for the bounds that we obtain, as one can run sublinear cost local computation over sublinear number of rounds). However, the algorithms  still follow  the CONGEST model (i.e., there is bandwidth restriction).

We make a note on the output of our distributed algorithms: at the end, each node will know which of its incident edges belong to the HC (exactly two of them).

% For any node $u$, $d(u)$ and $N(u)$ denote the degree of $u$ and the set of neighbors of $v$ in $G$ respectively.

\begin{figure}[t]
  \label{fig:dhc1}
  \centering
  \begin{subfigure}[b]{0.4\linewidth}
    \includegraphics[width=\linewidth]{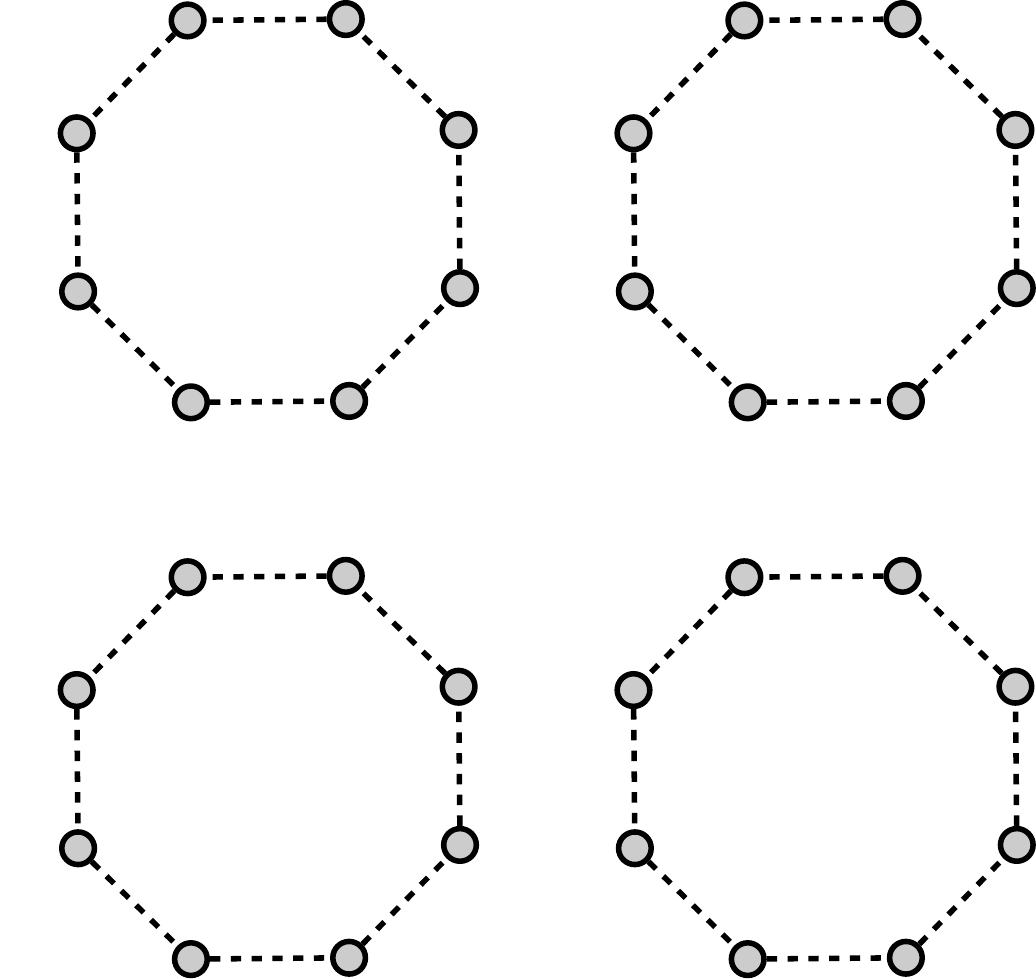}
    \caption{Phase1}
  \end{subfigure}
  \vline width 1pt
  \begin{subfigure}[b]{0.4\linewidth}
    \includegraphics[width=\linewidth]{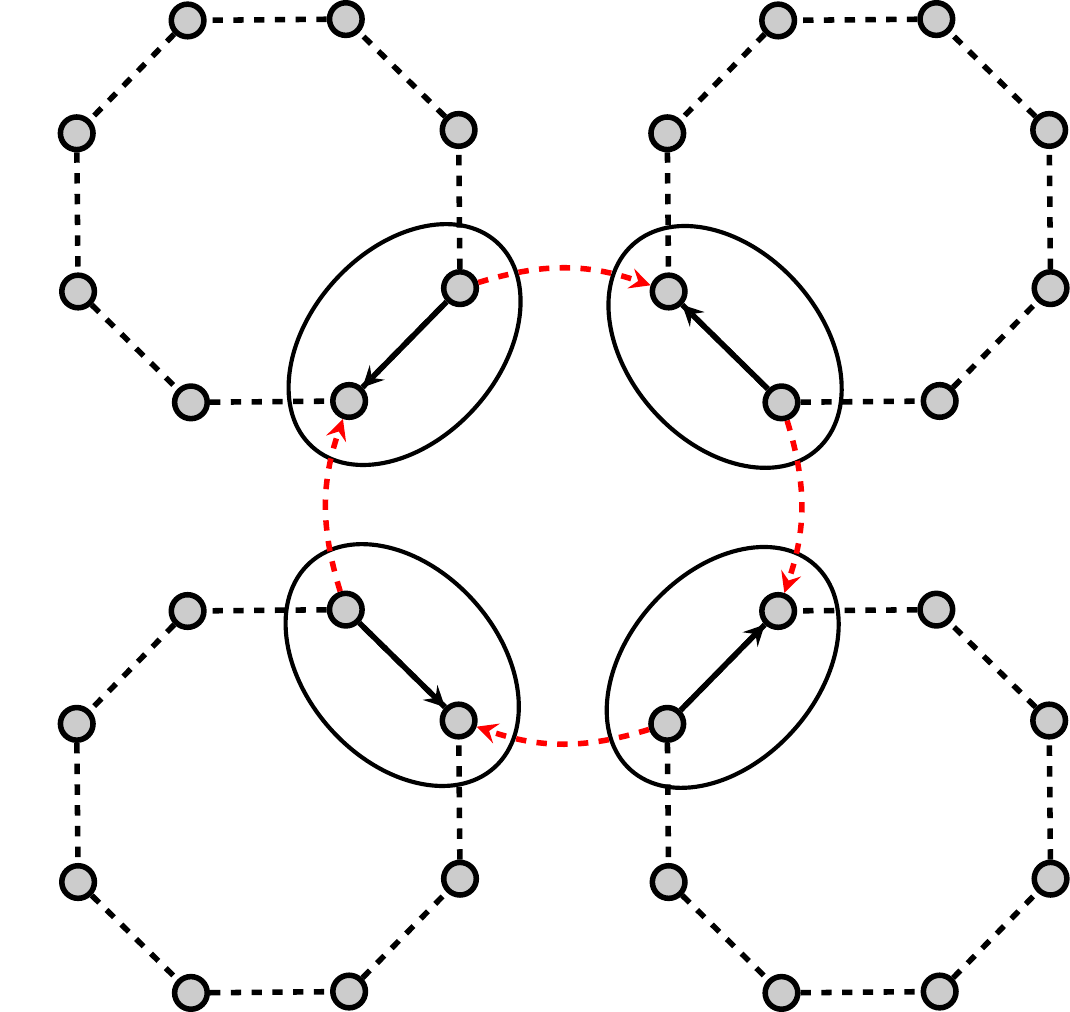}
    \caption{Phase2}
  \end{subfigure}
  \caption{Algorithm DHC1 builds HC in two phases. Phase 1 constructs $\sqrt{n}$ sub HCs in parallel. Phase 2 combines all sub HCs by building a HC over the graph of hyper nodes.}
  \label{fig:hc}
\end{figure}

\subsection{Other Related Work} \label{sec:related}

There are several algorithms for finding a HC in random graphs (both $G(n,p)$ and its closely related variant $G(n,M)$ random graphs), e.g., we refer to the survey due to Frieze \cite{frieze1988finding}.
There also have been work on parallel algorithms for finding Hamiltonian cycles in $G(n,p)$ random graphs.  Frieze \cite{frieze1987parallel} proposed two algorithms for EREW-PRAM machines: the first uses $O(n \log n)$ processors and runs in  $O(\log^2 n)$ time, while the second one uses $O(n \log^2 n )$ processors and runs ins $O((\log \log n)^2)$ time. MacKenzie and Stout \cite{mackenzie1993optimal} gave an algorithm for Arbitrary CRCW-PRAM machines that operates in $O(\log^* n)$ average time and requires $O(n/\log^* n)$ processors. All these parallel algorithms assume $p$ is a constant.

With regard to distributed algorithms, as mentioned earlier, the only prior work we are aware of is the work of Levy et al.\cite{levy2004distributed} which  gives a fully distributed algorithm to find a HC in
$O(n^{\frac{3}{4} + \epsilon})$ time when $p  =  \omega(\frac{\sqrt{\log{n}}}{n^{\frac{1}{4}}})$.
Their algorithm (based on the algorithm of MacKenzie and Stout \cite{mackenzie1993optimal}) works in three
phases: finding an initial cycle, finding $\sqrt{n}$ disjoint paths, and finally patching paths
into the cycle to build the  HC. Our fully distributed algorithms (Section \ref{sec:ouralg}) follow a different and a simpler approach and are significantly faster, while working for all ranges of
$p$ above the HC threshold.

\section{Fully-Distributed Algorithms}
\label{sec:ouralg}
In this section, we give two fast, efficient, fully-distributed algorithms
for the Hamiltonian cycle problem.
The first algorithm, in Section \ref{sec:sqrtn}, is a distributed algorithm for the  case of
$p = \frac{c\ln n}{\sqrt{n}}$ (throughout, $c$ will be a large enough constant, say bigger than 54)
and runs in time $\tilde{O}(\sqrt{n})$ rounds whp. (In fact, the algorithm will work for any $p \geq
\frac{c\ln n}{\sqrt{n}}$, but for simplicity we will fix $p = \frac{c\ln n}{\sqrt{n}}$.)
This algorithm works in the CONGEST model and is fully distributed, i.e., each node's local
computation memory is $o(n)$ and the computation cost per node per round is also $o(n)$.
This algorithm is somewhat simpler, contains some of the main ideas, and is also useful in  understanding the second algorithm.
The second algorithm, in Section \ref{subsec:algp}, is more general, and works for $p = \frac{c\ln n}{n^{\delta}}$, for any $0 < \delta \leq 1$ and runs in $\tilde{O}(n^{\delta})$ rounds. Both algorithms have two phases; while the first phase is similar for both algorithms, the second phase for the second algorithm is more involved.

Before we go into the details of our algorithms, we will give the main intuition.
Our algorithm is inspired by the well-studied {\em rotation} algorithm (the rotation is a simple operation described in Section \ref{sec:sqrtn}) that was used by Angluin and Valiant to develop a fast sequential algorithm for the $G(n,p)$ random graph for $p \geq \frac{c\ln n}{n}$ (for some suitably large constant $c$, say $c > 36$). However, this algorithm seems inherently sequential, since it tries to extend the cycle one edge at a time; hence the running time under this approach is at least $\Omega(n)$. To get a sublinear time, we follow a two-phase strategy which works in somewhat denser graphs, i.e., $p =  \frac{c\ln n}{n^{\delta}}$, for any $0 < \delta < 1$. In Phase 1, we partition the graph into {\em disjoint} random subgraphs each of size
(approximately) $\Theta(n^{\delta})$ (there will be $\Theta(n^{1-\delta})$ subgraphs). The intuition behind this partition is that
each subgraph will have a HC  of its own (of length equal to the size of the subgraph) whp, since it satisfies the threshold for Hamiltonian cycle (note that $p = \frac{c\ln n}{n^{\delta}}$). We use a distributed implementation of the rotation algorithm to find the Hamiltonian (sub)cycles independently in each of subgraphs --- this takes time essentially linear in the size of the subgraphs, i.e., $\tilde{O}(n^{\delta})$. In Phase 2, we stitch the cycles without taking too much additional time, i.e., in $\tilde{O}(n^{\delta})$ time. When $p = \frac{c\ln n}{\sqrt{n}}$,
the case is special, since the number of subgraphs and the size of each subgraph is balanced, so the stitching can be done by essentially implementing a modification of Phase 1 as follows.
Take two adjacent nodes from each subgraph cycle and find a Hamiltonian cycle between the chosen nodes (this has to be done carefully, so that it can be combined with the subgraph cycles to form a HC over all the nodes). Since $p =  \frac{c\ln n}{\sqrt{n}}$, and the number of chosen nodes is $\Theta(\sqrt{n})$, whp a HC exists between the chosen nodes and we can find it using a strategy similar to Phase 1.
For general $p$,
we note, that we cannot just simply stitch
as described above, since $p$ is much smaller than the needed threshold.
Hence,
we do the stitching in stages, as described in Section \ref{subsec:algp}.

\begin{figure}[t]
  \centering
  \includegraphics[width=\linewidth]{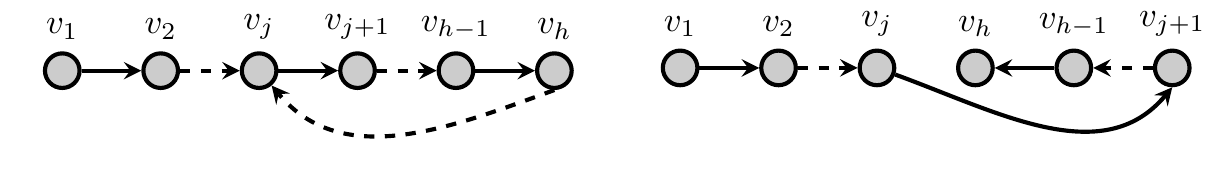}
  \caption{Path Rotation: Extending from the head node ($v_h$), we encounter a node ($v_j$) on the path. The right side shows  the rotated path.
}
\label{fig:rotate}
\end{figure}

\subsection{The Algorithm for $p = \frac{c\ln n}{\sqrt{n}}$}
\label{sec:sqrtn}
Our first algorithm, called the {\em Distributed Hamiltonian Cycle Algorithm 1  (DHC1)}, works  for a random graph $G(n,p=\frac{c\ln n}{\sqrt{n}})$, where $c$ is a suitably large constant.
%For the rest of this section, we discuss the algorithm. The proof for sublinear runtime will be in a separate session on analysis.

\subsubsection{High-Level Description of DHC1}
Given a random graph $G(n,\frac{c\ln n}{\sqrt{n}})$, our algorithm works in two phases. In Phase $1$, the graph is
partitioned into $\sqrt{n}$ subgraphs $G_i$, each of $\Theta(\sqrt{n})$ nodes. Then each subgraph
constructs its own Hamiltonian cycle $C_i$, independently in parallel. In Phase 2, the algorithm  finds a Hamiltonian cycle
connecting $C_1,\cdots,C_{\sqrt{n}}$. This is done as follows: for
each $C_i$,  pick only one edge $e_i = (v_i, u_i)$, call this a hypernode (edges inside oval shapes in Figure \ref{fig:hc}). Consider the graph
$G'$ of $\sqrt{n}$ hypernodes $e_i$, a hypernode uses $u_i$
as the \textit{incoming} port, and $v_i$ as the \textit{outgoing} port. In other words, we only
look at the edges $(v_j, u_i)$ and $(v_i, u_j)$ for any pair $e_i \neq e_j$. The algorithm constructs a
Hamiltonian cycle in $G'$ which is easy to see completes the Hamiltonian cycle in $G$ (See Figure \ref{fig:hc}).

In Phase 1 (as well as in Phase 2, for constructing a HC in $G'$), the cycles are constructed locally: each node becomes aware
of its predecessor and successor after the construction.
For convenience, each node also maintains an index
of its position in the cycle.
The resulting Hamiltonian cycle is hierarchical. Each node maintains
its index $subcyc$ in the subgraph cycle.
In Phase 2, if a node is part of a hypernode, it maintains an extra
index $hypcyc$ in the cycle constructed in Phase 2. When traversing the cycle, if a node has a $hypcyc$ link, follow it, otherwise follow the $subcyc$ link.

We next describe the distributed algorithm for constructing a HC in the $\sqrt{n}$-sized subgraph.
This distributed algorithm which we call {\em Distributed Rotation Algorithm (DRA)} is based on the well-known  randomized algorithm for finding a Hamiltonian cycle that uses so called {\em rotation} steps \cite{mitzenmacher2017probability}  (See Figure \ref{fig:rotate}).
%Notice that ``sequential'' means: sequential in cycle construction, the algorithm itself is distributed.
%We call this algorithm the Randomized Distributed Hamiltonian Cycle (RDHC) algorithm.
%which is used also in the sequential algorithm of Angluin and Valiant
%In essence, it has a low probability of failure, depending
%on the density of the graph. This is critical, as we would want the PDHC algorithm to
%have a low probability of failure as well.

\subsubsection{The DRA algorithm}
Consider a graph $G$ with $n$ nodes. We construct a Hamiltonian path $v_1,v_2,\cdots,v_n$;
if there is an edge connecting $v_n$ and $v_1$, then we have a Hamiltonian cycle. We will grow the
path sequentially by a simple randomized algorithm. For a path $v_1,\cdots,v_h$, let $v_h$ be the
\textit{head}. Initially, we choose a random $v_1$ which is also the initial \textit{head}. The
\textit{head} picks a random edge $(v_h,u)$, say, which has not previously been used.

If $u \not \in \{v_1, \cdots, v_h\}$, add node $u$ to the path and set it as the new \textit{head}.
If $u$ is some $v_j$, then we \textit{rotate} the path: $v_1,\cdots,v_j,v_{j+1},\cdots,v_h$ becomes
$v_1,\cdots, v_j, v_h, v_{h-1},\cdots, v_{j+1}$ and $v_{j+1}$ is the new head. The rotation can be
implemented by  just a renumbering: for $v_i$, where $j+1 \leq i \leq h$, reassign $i \gets h + j + 1 - i$.
In a distributed setting,
we can implement an efficient procedure: $v_j$ broadcasts the values $h$ and $j$ then every node
can renumber itself accordingly. Notice that the required time for broadcast is the diameter
$D$ of the graph, and we will give bounds for $D$ in the analysis.
\begin{algorithm*}[tbh]
\caption{Distributed Rotation Algorithm (DRA) Algorithm}
\label{alg:rdhc}
\begin{algorithmic}[1]
  \Function{DRA}{$G(V,E), cycindex$} \Comment{code for each node $v \in V$,
                                      use $cycindex$ for path index}
    \Init
      \State $v.unused \gets$ all edges to neighbors
      \State $v.cycindex \gets 0$
      \State only one $v$ becomes $head$, $v.cycindex \gets 1$
    \EndInit
    \While {$v.unused \neq \varnothing$}
      \If {$v$ is $head$}
        \State $(v,u) \gets$ random edge from $v.unused$
        \State $v.unused \gets v.unused - \{(v,u)\}$
        \State send to $u$: $progress(pos = v.cycindex)$
      \EndIf
      \OnReceive \ message $progress(pos)$
        \State \textbf{if} $pos=|V|$ and $v.cycindex=1$ \textbf{then} return $Success$
        \State $v.unused \gets v.unused - \{(sender,v)\}$
        \If {$v.cycindex = 0$}	\Comment {first time visiting $v$}
          \State become $head$: $v.cycindex \gets pos + 1$
        \Else	\Comment {$v$ is already on the path}
          \State broadcast: $rotation(h = pos, j = v.cycindex)$
        \EndIf
      \EndReceive
      \OnReceive \ message $rotation(h,j)$
        \If {$j<v.cycindex\leq h$}
          \State $v.cycindex \gets h + j + 1 - v.cycindex$
          \If {$v.cycindex = h$}
            \State $v$ becomes $head$
          \EndIf
        \EndIf
      \EndReceive
    \EndWhile
    \Return
  \EndFunction
\end{algorithmic}
\end{algorithm*}

The Distributed Rotation Algorithm (DRA)  is given in Algorithm \ref{alg:rdhc}, where we initialize the algorithm
by assigning any one node to be the $head$. The DHC1 algorithm pseudocode is given in Algorithm \ref{alg:pdhc}.
Notice that we initialize $2$ position indexes for each node, and construct the (overall)
Hamiltonian cycle by multiple calls to Algorithm \ref{alg:rdhc}.

\begin{algorithm*}[tbh]
\caption{Distributed Hamiltonian Cycle Algorithm 1 (DHC1)}
\label{alg:pdhc}
\begin{algorithmic}[1]
  \Function{DHC1}{$G(V,E)$}
    \Init
      \State $n \gets |V|$
      \foreach \ $v \in V$: set $v.subcyc$ and $v.hypcyc$ to $0$
      \endForeach
    \EndInit
    \Phase \ $1$
      \State $v.color \gets random[1,\cdots,\sqrt{n}]$
      \State $G_i(V_i,E_i)$ is a subgraph with nodes in color $i$
      \foreach \ $G_i$:
        \State $C_i \gets$ \Call{DRA}{$G_i, cycindex = subcyc$}
      \endForeach
    \EndPhase
    \Phase \ $2$
      \foreach \ $C_i$:
        \State pick a random $u_i \in C_i$
        \State $v_i \gets predecessor(u_i)$
        \State $hypernode_i \gets [u_i, v_i]$
        \State $G'$: graph of all $hypernode_i$, edges: all pairs $(v_j,u_k), j \neq k$
      \endForeach
      \State $C' \gets$ \Call{DRA}{$G', cycindex = hypcyc$}
    \EndPhase
    \State \Return
  \EndFunction
\end{algorithmic}
\end{algorithm*}

\subsubsection{Analysis}
We first state the main theorem, which gives the probability of success and the expected runtime
of the DHC1 algorithm.
%Then we will look at the supporting lemmas to complete the proof.
\begin{theorem}
  \label {the:pdhc}
  For a $G(n,p)$ with $p = \frac{c \ln n}{\sqrt{n}}$ with $c \geq 86$, the DHC1 algorithm successfully
  builds a Hamiltonian cycle with probability $(1  -  O(\frac{1}{n}))$, in $O(\sqrt{n}
  \frac{\ln^2{n}}{\ln{\ln{n}}})$ rounds.
\end{theorem}

The next theorem describes the performance of the Distributed Rotation  Algorithm (DRA), a key subroutine of the DHC1 algorithm. This result
will be used in both Phase 1 and Phase 2 of DHC1 to bound its runtime. To simplify the analysis, we will state the run time in this theorem
in terms of the number of \textit{steps}, where each step is one rotation or growing the path by one node. For the moment, we ignore the cost of broadcast, which we will later account for in the main theorem.

\begin{theorem}
  \label{the:seq}
  Given a $G(n,p)$ graph where $p \geq 86 \frac{\ln n}{n}$, the DRA algorithm
    constructs a Hamiltonian Cycle in $7n\ln{n}$ steps with probability of success $1  -  O(\frac{1}{n^3})$.
\end{theorem}

\begin{proof}
  We follow the approach as described in \cite{mitzenmacher2017probability} which we refer to for more details. The main idea is to relate the algorithm to a {\em coupon collector} process, where the goal is to collect $n$ different coupons and in each step the probability of collecting a particular coupon
  is $1/n$ (independent of other coupons) and it is known that all coupons can be collected in $O(n \ln n)$ steps whp.  Here, the $n$ coupons represent the $n$ nodes and collecting all the  coupons is analogous  to building a HC. Since the rotation algorithm does not give
  uniform $1/n$ probability, to apply the coupon collector model, we relax the analysis as follows.

  Considered a relaxed  algorithm such that every node  has equal probability of $\frac{1}{n}$ to
  be chosen in every step of growing the path (this relaxation is described in \cite{mitzenmacher2017probability}). Note that, in fact, the algorithm is more efficient
  in choosing	a new node. We will not restate all the details here, except for the key technique. Remember that the edge probability is $p$, and this implies a dependency between two nodes. Under the relaxed algorithm, let each node have a list of edges, called ``unused'' edges, which is selected independently at random, with probability $q$. The technical part is how to convert $p$ to $q$, such that the ``unused'' edges is a subset of the true edges. All the subtleties can be found in \cite{mitzenmacher2017probability}, for convenience, we cite $q$ here: $q = 1 - \sqrt{1-p} \geq p/2$. We are now ready for the proof, where we want to improve the analysis of \cite{mitzenmacher2017probability}. In particular, by allowing larger runtime, but still in $O(n\ln n)$, we can reduce the failure probability to $O(1/n^3)$. This technique can be extended to achieve failure probability in $O(1/n^\alpha)$, with a given constant $\alpha$.

  The relaxed algorithm has two scenarios of failure:
  \begin{itemize}
    \item $\mathcal{E}_1$: The algorithm runs for $7n\ln n$ steps while no unused edges in any
    vertex becomes empty, and fails to construct a Hamiltonian cycle.
    \item $\mathcal{E}_2$: At least one vertex runs out of unused edges during $7n\ln{n}$ steps.
  \end{itemize}
  For event $\mathcal{E}_1$, equal probability of $1/n$ gives: the probability of not seeing a node
  after $4 n \ln n$ steps is:
  \[ \left(1-\frac{1}{n} \right)^{4n \ln n} \leq \frac{1}{n^4}.
  \]
  Using union bound, the probability of failure to meet all $n$ nodes after $4n\ln{n}$ steps is:
  $O \left ( \frac{1}{n^3}\right )$.

  Now, in order to close the cycle, the head needs to visit the tail, which happens with probability
  $\frac{1}{n}$. After $3n \ln n$ steps, the probability of failure to complete the cycle is at most:
  \[
  \left(1-\frac{1}{n} \right)^{3n \ln n} \leq \frac{1}{n^3}.
  \]
  In total, $Pr(\mathcal{E}_1) \leq \frac{2}{n^3} = O\left(\frac{1}{n^3}\right)$.

  For event $\mathcal{E}_2$, we break it into two sub events:
  \begin{itemize}
    \item $\mathcal{E}_{2.1}$: At least $21 \ln n$ edges are removed from at least one node
    during $7n \ln n$ steps.
    \item $\mathcal{E}_{2.2}$: At least one node has fewer than $21\ln{n}$ edges in its initial
    unused list.
  \end{itemize}
  Consider $\mathcal{E}_{2.1}$ and look at a node $v$. Let $X$ be the number of edges removed at $v$ during $7n \ln n$ steps.
  We have $E[X]= \frac{1}{n} * 7n \ln n = 7\ln n$. Using Chernoff bound,
  \begin{align*}
    Pr(X\geq 21\ln n)) &= Pr(X \geq (1 + 2)7 \ln n) \\
    &\leq \left(\frac{e^2}{3^3} \right)^{7\ln n}  \leq \left( \frac{1}{e^{4/7}}\right)^{7\ln n}= O\left(\frac{1}{n^4}\right).
  \end{align*}
  Using union bound, $Pr(\mathcal{E}_{2.1}) = O\left(\frac{1}{n^3} \right)$.

  Consider $\mathcal{E}_{2.2}$.
  Let $Y$ be the initial number of edges in the unused edges list of a node.
  We have $E[Y] = q(n-1) \geq (43 \frac{\ln n}{n} ) (n-1) \geq 42 \ln n$. Using Chernoff's bound:
  \begin{align*}
    & Pr(Y \leq 21\ln n) = Pr(Y\leq (1 - \frac{1}{2})42 \ln n) \\
    & \leq \exp \left(-\frac{\left( \frac{1}{2}\right)^2 42\ln n}{2} \right) =
    O\left( \frac{1}{n^4}\right).
  \end{align*}

  Using union bound for $n$ nodes, $Pr(\mathcal{E}_{2.2}) = O\left(\frac{1}{n^3} \right)$.

  Union over the failure events, the failure probability is less than:
    $Pr(\mathcal{E}_1)+Pr(\mathcal{E}_{2.1}) + Pr(\mathcal{E}_{2.2}) = O(\frac{4}{n^3})$.
\end{proof}

Having analyzed the DRA algorithm, we return to the discussion
of our DHC1 algorithm.

\textbf{Analysis of Phase $1$:} Each subgraph $G_i$ uses the DRA algorithm
to independently construct (in parallel) its Hamiltonian cycle $C_i$. Because each subgraph performs the algorithm independently,
this phase is fully parallelized, and the expected runtime will be the expected runtime of
the largest subgraph. For the failure probability, we can simply use a union bound. We state the
following theorem for  Phase 1.

\begin{lemma}
  \label{lem:ph1}
  For a $G(n,p)$ with $p\geq \frac{c \ln n}{\sqrt{n}}$ where $c\geq 86$, Phase 1 of the algorithm
  succeeds with probability $1-O(1/n)$, in $O(\sqrt{n}\ln{n})$ steps.
\end{lemma}

To prove Lemma \ref{lem:ph1}, we will show that each partition has size in $\Theta(\sqrt{n})$ and
is sufficiently dense for the success of the DRA algorithm. In particular, we
introduce the following:
\begin{definition}
  Let $\mathcal{A}$ be the event that all partitions have size $a\sqrt{n}$, where
  $a\in[\frac{1}{2},\frac{3}{2}]$.
\end{definition}

\begin{lemma}
  \label{lem:partsize}
  DHC1 algorithm in Phase 1 (line 5)  partitions nodes such that event $\mathcal{A}$ happens with probability at least $1-O(\frac{1}{n})$.
  \begin{proof}
    Consider any single color. Let $X$ be a random variable representing the number of nodes with that
    color. Let $X_i, i=1,\cdots, n$ be indicator random variables of values $0,1$: $X_i = 1$ if
    node $i$ choses that color, $X_i=0$ otherwise. By linearity of expectation, we have $E[X] =
    E[\sum X_i] = \sum E[X_i] = n \frac{1}{\sqrt{n}} = \sqrt{n}$.

    In order to show that $X$ is concentrated around its expectation, $\frac{1}{2} E[X] \leq
    X \leq \frac{3}{2} E[X]$, we apply Chernoff bound:
    \begin{align*}
      Pr(|X - \sqrt{n}|\geq \frac{1}{2} \sqrt{n}) \leq 2e^{\frac{- (\frac{1}{2})^2\sqrt{n}}{3}} =
      2e^{\frac{- \sqrt{n}}{12}}.
    \end{align*}
    With $\sqrt{n}$ partitions, by union bound, we have:
    \begin{align*}
      Pr(\neg \mathcal{ A}) \leq \sqrt{n} \times 2e^{\frac{- \sqrt{n}}{12}} = O\left(\frac{1}{n}\right).
    \end{align*}
  \end{proof}
\end{lemma}

\begin{lemma}
  \label{lem:subhc}
  When event $\mathcal{A}$ happens, Phase 1 succeeds with probability $1 - O(\frac{1}{n})$.
\begin{proof}
  By Lemma \ref{lem:partsize}, each partition has size of $a\sqrt{n}$, where $\frac{1}{2} \leq a
  \leq \frac{3}{2}$. Consider a partition with $n'$ vertices as a random graph with probability $p'$.
  It is easy to show that $p' \geq 86 \ln n'/n'$, as follows.
  The probability for the presence of an edge
  in this partition is the same as in the original graph. We have:
  \begin{align*}
    p' &= p \geq 86\frac{\ln n}{\sqrt{n}} = 86 \frac{\ln \frac{{n'}^2}{a^2}}{\frac{n'}{a}}
    = 86 a \frac{2 \ln n' - \ln a^2 }{n'}.
  \end{align*}
  When $1/2 \leq a <1$, then $p' \geq 86 a \frac{2 \ln n'}{n'} \geq 86 \frac{\ln n'}{n'}$.\\
  When $1 \leq a \leq 3/2$, then $p' \geq 86 a \frac{2 \ln n'}{2an'} = 86
  \frac{\ln n'}{n'}$, using the fact that $x - y > \frac{x}{2z}$, for $x$ sufficiently large and
  small constants $y,z$ such that $z>1$.

  Applying theorem \ref{the:seq}, the probability of failing for this partition is
  $O(\frac{1}{(\sqrt{n})^3})$.
  Using union bound, the probability of failure in phase 1 is at most:
  $\sqrt{n} \times O(\frac{1}{(\sqrt{n})^3}) = O(\frac{1}{n})$.
\end{proof}
\end{lemma}

\begin{proof}[Proof of Lemma \ref{lem:ph1}]
  By Lemma \ref{lem:subhc} and Lemma \ref{lem:partsize}, the probability of failure for phase $1$ of is
  $O(\frac{1}{n})$.

  For the runtime of this phase, we apply Theorem \ref{the:seq}. Consider a partition of size $a\sqrt{n}$,
  the runtime is: $7a\sqrt{n}\ln(a\sqrt{n})$. Each partition executes Algorithm \ref{alg:rdhc} in
  parallel, the runtime is dominated by the largest partition. Sine $a \leq \frac{3}{2}$, the
  runtime of phase 1 is: $O(\sqrt{n}\ln{n})$.
\end{proof}

\textbf{Analysis of Phase $2$:} In this phase, we apply the DRA algorithm on the $G'$ graph of
hypernodes. We only need to show that $G'$ is dense enough to apply Theorem \ref{the:seq}. We have
the following lemma.

\begin{lemma}
  \label{lem:ph2}
  For a $G(n,p)$ with $p\geq \frac{c \ln n}{\sqrt{n}}$ where $c\geq 86$, Phase 2 of the DHC1
  algorithm succeeds with probability $O\left(1 - \frac{1}{n^{\frac{3}{2}}} \right)$, in
  $O(\sqrt{n}\ln{n})$ steps.
\begin{proof}
  The graph $G'$ constructed according to the algorithm is a random graph with $n'=\sqrt{n}$
  and the edge probability $p'$. Consider a pair $(e_i=[v_i, u_i], e_j=[v_j,u_j])$ of
hypernodes, by construction, the probability to have an edge between them is: $p' = 1 - (1-p)^2 \geq
p$, where $p$ is the probability for an edge between two nodes in the orginial $G$ graph.
  \[
    p' \geq p \geq 86 \frac{\ln{n}}{\sqrt{n}} > 86 \frac{\ln n'}{n'}.
  \]
  Applying Theorem \ref{the:seq} this phase succeeds with probability $O\left(1 - \frac{1}{n^{\frac{3}{2}}}
  \right)$ in $O(\sqrt{n}\ln n)$ steps.
\end{proof}
\end{lemma}

\begin{proof}[Proof of Theorem \ref{the:pdhc}]
  The proof of the main theorem then follows trivially, by Lemma \ref{lem:ph1} and Lemma \ref{lem:ph2}.
  The probability of success is:
  \[
    O\left( 1 - \frac{1}{n} \right) O \left(1 - \frac{1}{n^{3/2}}\right) =
    O\left( 1 - \frac{1}{n} \right).
  \]
  The number of steps in each phase is: $O(\sqrt{n}\ln n)$.
  In the worst case, consider we have broadcast in every step, then, the number of rounds is the
  number of steps multiplied by $O(D)$ where $D$ is the diameter of the graph executing the DRA
  algorithm. In both Phase 1 and Phase 2, the graphs are random graphs under the model $G(n',p')$
  where $p'\geq 86 \ln n' / n'$, and $n' = \Theta(\sqrt{n})$. By \cite{chung2001257}, the diameter
  of these graphs is $\Theta(\frac{\ln n'}{\ln \ln n'}) = \Theta(\frac{\ln n}{\ln \ln n})$.

  Therefore, the number of rounds is bounded by:
  \[
    O\left( \sqrt{n} \frac{(\ln n)^2}{\ln \ln n} \right).
  \]
\end{proof}

\iffalse
\subsection{Message Complexity}
\begin{theorem}
  The expected message complexity of PDHC algorithm \ref{alg:pdhc} is $O(n^{\frac{3}{2}}\ln^2{n})$.
\end{theorem}
\begin{proof}
  The PDHC algorithm takes two phases.
  \par
  \textbf{Phase 1:} Consider each sub graph $G_i$ does broadcast at each step. Phase 1 takes $7/2 \sqrt{n} \ln n$ steps and the broadcast must propagate on all inner-subgraph edges. The expected number of all inner-subgraph edges is $\binom{n}{2} p \frac{1}{\sqrt{n}}= O(n\ln n)$. Therefore the expected number of messages for phase $1$ is $O(7/2 \sqrt{n} \ln n \times n \ln n)=O(n^{3/2} (\ln n)^2)$.
  \par
  \textbf{Phase 2:}
  Consider the hyper-graph $G'$. Phase 2 takes $7/2 \sqrt{n} \ln n$ steps and the broadcast must propagate on all edges in $G'$. The expected number of edges in $G'$ is $2\binom{\sqrt{n}}{2} p = O(np)$. Therefore, the expected message complexity for phase $2$ is $O(np \times \sqrt{n} \ln n) =  O(n (\ln n)^2))$.
  \par
  In total, the expected message complexity is $O(n^{\frac{3}{2}}\ln^2{n})$
\end{proof}

\subsection{Hamiltonian Cycle Verification}
We use a distributed MST algorithm in \cite{ghaffari2017distributed} in order to verify a Hamiltonian cycle. First, we set weight zero for each edge on Hamiltonian cycle and none zero value for other edges. If MST algorithm builds a minimum spanning tree successfully having sum of the edges zero, we check whether the two end nodes are connected by an edge of weight zero ro not (an edge belong to Hamiltonian cycle). It it was, our verification test returns True.

\fi

\subsection{The Algorithm for $p = \frac{c\ln n}{n^\delta}$}
\label{subsec:algp}

\begin{algorithm*}[tbh]
  \caption{Distributed Hamiltonian Cycle Algorithm 2 (DHC2). Code for $v \in G(V,E)$.}
  \label{alg:hhc}
  \begin{algorithmic}[1]
    \Phase \ $1$
    \State Run phase 1 of algorithm \ref{alg:pdhc}, using $n^{1-\delta}$ colors
    \EndPhase
    \Phase \ $2$
      \For {$i = 1 \cdots \lceil \log n^{1-\delta} \rceil $}
        \If {$v.color$ is $odd$} \Comment{$v$ is an active node}
          \State send message $\mathit{verify}(succ(v))$ to all its neighbors with color $v.color+1$
          \OnReceive \ $\cup\{ \mathit{verified}(u,u') \}$
            %\State Select the smallest $(u,u')$, construct candidate bridge: $candidate \gets ((v, succ(v)),(u,u'))$
            \State Select the smallest $(u,u')$, construct candidate bridge: $candidate \gets ((v, u'),(u,succ(v)))$
            \State Broadcast $candidate$ within $v$'s partition
            \If {$candidate = min(\cup candidates)$}
              \State Send message $buildBridge$ to $u$
              \State Broadcast $Renumbring$ inside HC
            \EndIf
          \EndReceive
        \EndIf
        \OnReceive \ message $\mathit{verify}(u)$ \Comment{only passive nodes receive this type message}
            \State ask $succ(v)$ and $pred(v)$ if they have $u$ as their $(v.color - 1)$ neighbor
            \State if $succ(v)$ (or $pred(v))$ confirmed, set $u'$ to $succ(v)$ (or $pred(v)$) , reply to sender: $\mathit{verified}(v, u')$
          \EndReceive
          \OnReceive \ message $buildBridge$
          \State Broadcast $Renumbring$  HC
          \EndReceive
        \State $v.color \gets \lceil v.color/2 \rceil$
      \EndFor
    \EndPhase
  \end{algorithmic}
\end{algorithm*}

We proved that for a $G(n,p)$ with $p=\frac{c \ln n}{\sqrt{n}}$, the DHC1 algorithm \ref{alg:pdhc} finds
a Hamiltonian cycle in $\tilde{O}(\sqrt{n})$ times. It is natural to ask the question: what is the performance on sparser graphs? Consider a $G(n,p)$ random graph where
$p=O(\frac{c \ln n}{n^\delta})$, for any $\delta \in (0,1)$. If we divide the graph into $n^{1-\delta}$ partitions,
each of size $n^\delta$, then Phase 1 of the DHC1 algorithm will work.
However, Phase 2 will not, since the graph of hypernodes is too sparse, under the threshold
required for the presence of an Hamiltonian cycle in Phase 2.

We present a general algorithm \ref{alg:hhc} called DHC2 that finds a Hamiltonian cycle
in random graphs $G(n,p)$ where $p=O(\frac{c\ln n}{n^\delta})$, where $c$ is a suitably large constant. This algorithm also has two phases. Phase $1$ is essentially a generalization of Phase 1 of DHC1, with $n^{1-\delta}$ partitions.
In phase $2$ of DHC2, we recursively merge  pairs of two disjoint cycles (in parallel) until the final cycle is
formed. Figure \ref{fig:hhc} depicts these merging steps. It follows that the algorithm constructs the final Hamiltonian cycle if it always successes in merging. We will show that this probability is very
high. But let's first describe the merging procedure.

\begin{figure}[t]
  \label{fig:dhc2}
  \centering
  \includegraphics[width=\linewidth]{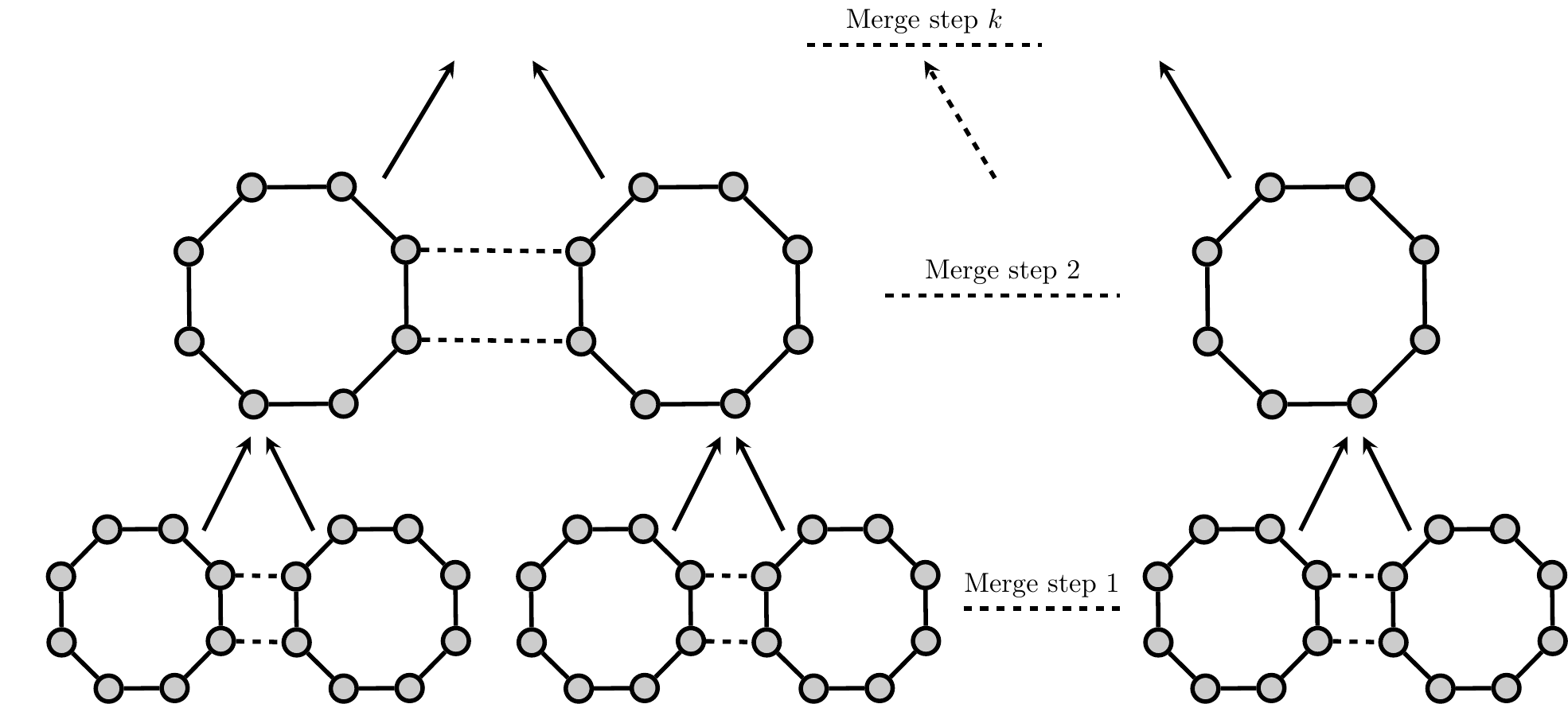}
  \caption{Phase 2 of the DHC2 algorithm: Merging pairs of cycles in a tree-like fashion. There are $O(\log n)$ merge steps, in each step, all HC pairs merge in parallel. The figure also shows how a pair of cycles are merged into a larger cycle by choosing two bridge edges.}
  \label{fig:hhc}
\end{figure}
%Gopal --- Show a figure that illustrates the merging of cycles --- a tree figure.

%Gopal --- Describe the idea of merging by taking two cycles and explain what a bridge is.
%Show a figure.

To merge the cycles, we define a rule for pairing them, then describe the merging by finding
a ``bridge'' between a pair of two cycles, as explained below.  Let's have the cycles indexed by
colors: $HC_1,HC_2,...,HC_{n^{1-\delta}}$. The pairing rule is to match two consecutive cycles, from left to right: $(HC_1,HC_2),\cdots,(HC_{2k+1},HC_{2k+2}),\cdots$, at most one cycle will be left out. Each pair merges independently in parallel, then every node (thus every cycle, including the left out one), updates their respective colors: $color \gets \lceil color/2 \rceil$. Therefore, the next merge step can progress with the same pairing rule, and every cycle is aware of its pair in all steps. It is clear that we need $\lceil \log(n^{1-\delta}) \rceil = O(\log(n))$ merge steps. To merge two cycles, we need to pick one ``bridge'' between them.
Let $e_i=(v_i,u_i) \in HC_i$ and $e_j=(v_j,u_j) \in HC_j$ where $(HC_i, HC_j)$ is a pair.
If there are two edges $(v_i, v_j)$ and $(u_i,u_j)$ or two edges $(v_i, u_j)$
and $(u_i,v_j)$ in $G(n,p)$, then we say $(e_i,e_j)$ is a bridge of $(HC_i, HC_j)$. The idea is, that each node
can check if it is part of a bridge, in parallel. Then within $HC_i$ and $HC_j$, each node broadcasts the discovered
bridge. This is done so as to choose one unique bridge per pair (since there may be more than one bridge per pair).  Each cycle choses the smallest bridge (say, based on the IDs of the bridge nodes).
Once a bridge is chosen, for example,  merging is done by each node independently
updating its $cycindex$, and updating $color$ (as mentioned above) for the next merging step. For
efficiency, in a pair, only the cycle with smaller $color$ will initiate the process, as shown in algorithm \ref{alg:hhc}.

Also, to avoid cluttering the algorithm \ref{alg:hhc}, we did not specify the renumbering process. This is trivial, given the bridge, and the size of the two cycles. Initially, each cycle performs a broadcast, so that its member nodes get to know the cycle size. Then, this information is attached to the bridge building message. From that onwards, every node can keep track of the size of the cycle that it is part of until the merging process is finished.

\begin{lemma}
  \label{lem:mhpdh2}
  Phase 1 of the DHC2 algorithm succeeds in $O(n^\delta \ln n)$ steps, with probability at least $1 - O(\frac{1}{n})$.
\begin{proof}
  Similar to Lemma \ref{lem:partsize}, it is easy to see that all partitions have size concentrated
  around the expected size, which is $\Theta(n^\delta)$. Consider a single color, let $X$ be the random variable
  of the size of the corresponding partition. Let $X_i$ be indicator random variables: $X_i=1$ if node $i$ choses
  that color, $0$ other wise. By linearity of expectation we have $E[X]=\frac{n}{n^{1-\delta}}=n^\delta$. Chernoff's bound gives:
  \[ Pr(|X-n^\delta)| \geq \frac{1}{2} n^\delta) \leq 2e^{-n^\delta / 12}. \]
  By union bound, all $n^{1-\delta}$ partitions have sizes in $\Theta(n^\delta)$, with probability:
  \[ O\left( n^{1-\delta} 2e^{-n^\delta / 12} \right) = O\left( \frac{1}{n} \right). \]
  Consider a partition with size: $n' = \Theta(n^\delta)$, as a random graph with edge
  probability $p'$. We have:
  \begin{align*}
    p' = p = O\left( \frac{\ln n}{n^\delta}\right) =
      O \left( \frac{1}{\delta} \frac{\ln n'}{n'} \right) = O \left( \frac{\ln n'}{n'} \right).
  \end{align*}
  By Theorem \ref{the:seq}, note  that we can reduce the probability of failure to
  $O\left( \frac{1}{n^{2-\delta}} \right)$ by increasing the number of steps by some factor of
  $(2-\delta)$. Thus, the number of required steps is: $O((2-\delta)n' \ln n') =
  O(n^\delta \ln n)$.

  Using union bound for $n^{1-\delta}$ partitions,
  the probability of failure is bounded above by:
  \[ n^{1-\delta} \times O\left( \frac{1}{n^{2-\delta}} \right) = O\left( \frac{1}{n} \right). \]
\end{proof}
\end{lemma}

To prove that Phase 2 of the DHC2 algorithm succeeds, we will first show the probability of
success for merge step $1$.

%Gopal --- Change HHC to DHC2 throughput.

\begin{lemma}
\label{lem:mhpdh3}
The merging of $\frac{np}{\ln n} = n^{1-\delta}$ Hamiltonian cycles in the first merging step of Phase 2
will be successful, with very high probability.

\begin{proof}
  Consider two partitions with two cycles $C, C'$, each with expected size $n^\delta$. Fix an
  edge $e$ in $C$, the probability that $e$ has a bridge to a fixed edge in $C'$ is at least $p^2$. Consider the set $S'$ of all non-adjacent edges in $C'$, such that $|S'|$ is maximal, The
  probability that $e$ does not have any bridge to $C'$ is at most the probability that $e$ does not have any bridge to $S'$:
  \begin{align*}
    (1 - p^2)^{n^\delta /2} &= O\left( \left( 1 -
    \frac{(\ln n)^2}{(n^\delta)^2}\right)^{n^\delta /2} \right) \\
    &= O \left( \left(e^{-(\ln n)^2} \right)^{1/(2\sqrt{n^\delta})} \right) \\
    &= O \left( n^{-\frac{\ln n}{2n^{\delta/2}}} \right).
  \end{align*}
  Consider the set $S$ of all non-adjacent edges in $C$, such that $|S|$ is maximal, the probability that all edges in $S$ has no bridge to $C'$ is:
  \begin{align*}
    O\left( \left( n^{-\frac{\ln n}{2n^{\delta/2}}} \right)^{n^\delta /2} \right) =
    O\left( n^{-n^{\delta/2}\ln n} \right).
  \end{align*}
  The above is the bound for the probability that $C$ and $C'$ fail to merge. We have $n^{1-\delta}/2$ pairs to merge,
  thus, union bound gives the failure probability:
  \begin{align*}
    \frac{n^{1-\delta}}{2} \times O\left( n^{-n^{\delta/2}\ln n} \right) =
    O\left( n^{-n^{\delta/2}\ln n + 1 - \delta} \right).
  \end{align*}
\end{proof}
\end{lemma}

\begin{lemma}
\label{lem:mhpdh4}
Phase 2 of the HHC algorithm is successful with very high probability which is $1 - o(1/n)$.
\begin{proof}
  Observe that after merging, the size of the Hamiltonian cycles increase, thus, in successive merge steps, the probability of failure becomes smaller than that in the first step.
  Using Lemma \ref{lem:mhpdh3}, with $O(\log{n})$ merge steps, union bound of the failure of Phase 2 is:
  \begin{align*}
    O\left( \ln{n} \cdot n^{-n^{\delta/2}\ln n + 1 - \delta} \right) = o\left(\frac{1}{n}\right).
  \end{align*}
\end{proof}
\end{lemma}

\begin{theorem}
  The DHC2 algorithm succeeds with probability $1 - O(\frac{1}{n})$ in $\tilde O(n^\delta)$
  steps.
\begin{proof}
  By Lemma \ref{lem:mhpdh2} and Lemma \ref{lem:mhpdh4}, the probability that the DHC2 algorithm
  succeeds is:
  \begin{align*}
    \left(1 - O\left(\frac{1}{n}\right)\right)\left(1 - o\left(\frac{1}{n}\right)\right)  =
    1 - O\left(\frac{1}{n}\right).
  \end{align*}
  To find the time complexity, we proceed similarly to the analysis of DHC1 algorithm: first
  calculate the number of steps, then consider the number of rounds required for broadcast.

  In Phase 1,
  the size of a subgraph is $n' = \Theta(n^\delta)$, and the edge probability is $p'=
  O(\ln n'/n')$, and by \cite{chung2001257}, the diameter is $O(\frac{\ln n'}{\ln \ln n'}) =
  O(\frac{\ln n}{\ln \ln n})$.

  In Phase 2, each merging takes constant number of rounds, and the broadcast time depends on the diameter of a subgraph. Observe that after each merging, we have a larger subgraph, while the edge probability is fixed, thus relative to the size, this larger subgraph is denser.
  Therefore we can bound the diameter of the merged subgraphs by the diameter of subgraphs in the first level, which is $O(\frac{\ln n}{\ln \ln n})$.

  The number of rounds for our DHC2 algorithm is then:
  \begin{align*}
    & O\left(n^\delta\ln n \frac{\ln n}{\ln \ln n} \right) +
    O\left(\ln n \frac{\ln n}{\ln \ln n} \right) \\
    &= O\left(\frac{n^\delta(\ln n)^2}{\ln \ln n} \right).
  \end{align*}

\end{proof}
\end{theorem}

\section{The Upcast Algorithm: A Centralized Approach}
In this section we consider what perhaps is the simplest and most obvious strategy of all --- we collect ``sufficiently large'' number of edges at some pre-designated root and then leave it to the root to compute a
Hamiltonian cycle.

\subsection{The Upcast Algorithm}\label{sec:upcastalg}
\begin{enumerate}
  \item Elect a leader, call it $v$. This step takes $O(D)$ rounds.

  \item Construct a BFS tree rooted at $v$, and call it $\mathcal{B}$. This step takes $O(D)$ rounds.

  \item All nodes except $v$ sample some $c'\log{n}$
\footnote{all the logarithms are natural logarithms}
of their adjacent edges (for a sufficiently large constant $c'$) --- independently and randomly --- and send the sampled edges to $v$ via the BFS tree constructed in the previous step.

  \item The root $v$ computes a Hamiltonian cycle locally and downcasts it to the rest of the nodes in $G$. This step takes essentially the same number of rounds as the previous (upcast) step.
\end{enumerate}

The main technical challenge in the analysis is showing that the upcast
can be done in time $\tilde{O}(1/p)$. This is done by showing that in a BFS tree in a random graph, the sizes of the subtrees rooted at every node are balanced (i.e., essentially the same size) whp. This ensures that the congestion
at each node during upcast is balanced and is $\tilde{O}(1/p)$.

\subsection{Analysis for the special case when $p = \Theta(\frac{\log{n}}{\sqrt{n}}).$}
Let $D$ be the diameter of $G = (V, E)$. Then Corollary 7 in \cite{bollobas1981diameter} implies that
\begin{fact}\label{fact-diameter-is-two}
  $D = 2$ when $p = \Theta(\frac{\log{n}}{\sqrt{n}})$.
\end{fact}

Thus Steps $1$ and $2$ take $O(1)$ time in total. We claim that Step $3$ in the algorithm takes $O(\sqrt{n}\log^2{n})$ rounds with high probability.\\

For $i \geq 0$, let $L_i$ be the nodes at level $i$ in the BFS tree $\mathcal{B}$. That is, $L_0 = \left\{v\right\}$, $L_1  =  \left\{w \in V\ |\ (v, w) \in E\right\}$, and $L_2  =
\left\{w \in V\ |\ dist(v, w) = 2\right\}$. We note that $L_0 \cup L_1 \cup L_2 = V$ by dint of Fact \ref{fact-diameter-is-two}.
\begin{lemma}\label{lemma-size-of-L1}
  $c(1 - \delta_1)(1 - \delta_2)\sqrt{n}\log{n}   \leq   |L_1|   \leq   c(1 + \delta_1)\sqrt{n}\log{n}$ with high probability for any fixed constants $\delta_1, \delta_2 \in (0, 1)$.
\end{lemma}
\begin{proof}
As $p = \frac{c\log{n}}{\sqrt{n}}$, $E[|L_1|] = (n - 1)p = \frac{c(n - 1)\log{n}}{\sqrt{n}} = c\sqrt{n}\log{n} - o(1)   \implies   (1 - \delta_2)c\sqrt{n}\log{n}  \leq  E[|L_1|]  \leq  c\sqrt{n}\log{n}$, for any fixed
constant $\delta_2$ in $(0, 1)$. A simple application of Chernoff bound gives us
\begin{align*}
  &Pr(|L_1|  \geq  c(1 + \delta_1)\sqrt{n}\log{n})\\
  &\leq  \text{exp}(-\frac{\delta_1^2 \cdot c(1 - \delta_2)\sqrt{n}\log{n}}{3})\\
  &=   n^{-\frac{\delta_1^2 \cdot c(1 - \delta_2)\sqrt{n}}{3}}\text{.}
\end{align*}
Similarly,
\begin{align*}
  &Pr(|L_1|  \leq  c(1 - \delta_1)(1 - \delta_2)\sqrt{n}\log{n})\\
  &\leq  \text{exp}(-\frac{\delta_1^2 \cdot c(1 - \delta_2)\sqrt{n}\log{n}}{2})\\
  &=   n^{-\frac{\delta_1^2 \cdot c(1 - \delta_2)\sqrt{n}}{2}}\text{.}
\end{align*}
\end{proof}
\begin{lemma}\label{lemma-size-of-L2}
  $n  -  (1 + c(1 + \delta_1)\sqrt{n}\log{n})   \leq   |L_2|   \leq   n  -  (1  +  c(1 - \delta_1)(1 - \delta_2)\sqrt{n}\log{n})$ with high probability for any fixed constants $\delta_1, \delta_2 \in (0, 1)$.
\end{lemma}
\begin{proof}
Follows directly from Fact \ref{fact-diameter-is-two} and Lemma \ref{lemma-size-of-L1}.
\end{proof}
For $w \in L_1$, let $\Gamma_{\mathcal{B}}(w)$ be the set of children of $w$ in the BFS tree $\mathcal{B}$. Then
\begin{lemma}\label{lemma-size-of-Gamma-w-in-B}
  $(1 - \delta_3)(n  -  (1 + c(1 + \delta_1)\sqrt{n}\log{n}))p   \leq   |\Gamma_{\mathcal{B}}(w)|   \leq   (1 + \delta_3)(n  -  (1  +  c(1 - \delta_1)(1 - \delta_2)\sqrt{n}\log{n}))p$ with high probability
  for any fixed constants $\delta_1, \delta_2, \delta_3 \in (0, 1)$.
\end{lemma}
\begin{proof}
  Similar to that of Lemma \ref{lemma-size-of-L1}.
\end{proof}
\begin{lemma}\label{lemma-size-of-Gamma-w-in-B-simplified}
  $c(1 - \delta_3)(1 - \delta_4)(1 - \delta_5)\sqrt{n}\log{n}   \leq   |\Gamma_{\mathcal{B}}(w)|   \leq   c(1 + \delta_3)(1 + \delta_4)\sqrt{n}\log{n}$ with high probability for any fixed constants
  $\delta_3, \delta_4, \delta_5  \in  (0, 1)$.
\end{lemma}
\begin{proof}
  Simplifying Lemma \ref{lemma-size-of-Gamma-w-in-B}.
\end{proof}
Since the ``high probability'' in Lemma \ref{lemma-size-of-Gamma-w-in-B-simplified} is actually exponentially high \footnote{that is $\geq 1 - \frac{1}{n^{\text{poly}(n)}}.$} (please refer to the proof of Lemma
\ref{lemma-size-of-L1}), we can take union bound over all $w \in L_1$, and get
\begin{lemma}\label{lemma-size-of-Gamma-w-final}
  The following statement holds with high probability: For all $w \in L_1$, $c(1 - \delta_3)(1 - \delta_4)(1 - \delta_5)\sqrt{n}\log{n}   \leq   |\Gamma_{\mathcal{B}}(w)|   \leq
  c(1 + \delta_3)(1 + \delta_4)\sqrt{n}\log{n}$ for any fixed constants $\delta_3, \delta_4, \delta_5  \in  (0, 1)$.
\end{lemma}
\begin{lemma}\label{lemma-upcast-number-of rounds}
  The upcast process takes at most $\frac{b}{\mathbb{B}} \cdot (c'\log{n}  +  cc'(1 + \delta_3)(1 + \delta_4)\sqrt{n}\log^2{n})$ rounds, where $\mathbb{B}$ is the bandwidth of the network, each edge is encoded in
  $b$ bits, and $0 < \delta_3, \delta_4 < 1$ are fixed constants.
\end{lemma}
\begin{proof}
  Follows directly from Lemma \ref{lemma-size-of-Gamma-w-final}.
\end{proof}
Usually we would have $b = \Theta(\log{n})$ and $\mathbb{B} = \Theta(\log{n})$, and that gives us the main result of this section ---
\begin{theorem}\label{theorem-centralized-upcast-algorithm}
  The Upcast algorithm  solves the distributed Hamiltonian Cycle problem in $G(n, p)$ random graphs in $O(\sqrt{n}\log^2{n})$ rounds, when $p  =  \Theta(\frac{\log{n}}{\sqrt{n}})$. Both
  the success probability and the running time hold with high probability.
\end{theorem}

\subsection{Analysis for the general case when $p = \Theta(\frac{\log{n}}{n^{1 - \epsilon}})$ for some constant $\epsilon \in (0, 1)$}
Let $D$ be the diameter of the graph $G = (V, E)$. Let $K$ be the smallest integer such that $K\epsilon \geq 1$, i.e., $K  \defeq  \lceil\frac{1}{\epsilon}\rceil$. Then Klee and Larman showed that \cite{Klee_1981}
\begin{fact}\label{fact-diamater-is-constant}
  $Pr(D(G) = K)  \rightarrow 1$ as $n \rightarrow \infty$, when $p = \frac{c\log{n}}{n^{1 - \epsilon}}$ for some positive constant $c$.
\end{fact}

Thus Steps $1$ and $2$ in the upcast algorithm take $O(1)$ time in total. We claim that Step $3$ takes $O(\frac{\log{n}}{p}) = O(n^{1 - \epsilon})$ rounds.\\

In a graph $G$, we denote by $\Gamma_k(x)$ the set of vertices in G at distance $k$ from a vertex $x$:
\begin{center}
  $\Gamma_k(x)  \defeq  \left\{y \in G\ |\ dist(x, y) = k\right\}$.
\end{center}
We define $\mathcal{N}_k(x)$ to be the set of vertices within distance $k$ of $x$:
\begin{align*}
  \mathcal{N}_k(x)  \defeq  \bigcup_{i = 0}^k	\Gamma_i(x)\text{.}
\end{align*}

We can adapt Lemma 3 in \cite{chung2001257} to show that
\begin{lemma}\label{lemma-BFS-tree-upper-bound}
  For any constant $\delta > 0$, with probability at least $1 - \frac{1}{n^3}$, we have
  \begin{enumerate}
    \item $|\Gamma_i(x)|  \leq  (1 + \delta)(np)^i$, $\forall 1 \leq i \leq D.$
    \item $|\mathcal{N}_i(x)|  \leq  (1 + 2\delta)(np)^i$, $\forall 1 \leq i \leq D.$
  \end{enumerate}
\end{lemma}

Lemma \ref{lemma-BFS-tree-upper-bound} basically says that the BFS tree $\mathcal{B}$ is essentially balanced. Hence an upcast algorithm would take $O(\frac{b}{\mathbb{B}} \cdot (1 + \delta)^D \cdot
\frac{n\log{n}}{d_v})$ rounds, where $\delta$ is any fixed positive constant, $n = |\mathcal{B}|$, and $d_v$ is the degree of the root $v$.
As $D = K = \lceil\frac{1}{\epsilon}\rceil$ is a constant, this implies a time complexity of $O(\frac{n\log{n}}{d_v})$. But $d_v$ is concentrated around $np$ with high probability. Thus an upcast algorithm would take
$O(\frac{n\log{n}}{np}) = O(\frac{\log{n}}{p})$ rounds with high probability. That is the main theorem of this section:
\begin{theorem}\label{theorem-upcast-time-generalized}
  The Upcast algorithm  solves the distributed Hamiltonian Cycle problem in $G(n, p)$ random graphs in $O(\frac{\log{n}}{p}) = O(n^{1 - \epsilon})$ rounds, when $p  =
  \Theta(\frac{\log{n}}{n^{1 - \epsilon}})$ for some constant $\epsilon \in (0, 1)$. Both the success probability and the running time hold with high probability.
\end{theorem}

We illustrate Theorem \ref{theorem-upcast-time-generalized} by describing an easy-to-understand (numerical) special case.

\begin{corollary}
  The Upcast algorithm solves the distributed Hamiltonian Cycle problem in $G(n, p)$ random graphs in $O(n^{\frac{2}{3}})$ rounds, when $p = \Theta(\frac{\log{n}}{n^{\frac{2}{3}}})$.
  Both the success
  probability and the running time hold with high probability.
\end{corollary}

\section{Conclusion}\label{sec:conc}

We present fast and efficient distributed algorithms for the fundamental Hamiltonian cycle problem in random graphs. Our algorithm (DHC2) is fully-distributed and runs in truly sublinear time ---
$\tilde{O}(\frac{1}{p})$ --- for all ranges of $p$; in fact,  denser the graph,  smaller the running time. We also present a conceptually simpler upcast algorithm with the same running time, but it is not
fully-distributed, and does \emph{not} achieve \emph{load-balancing}.

Our fully-distributed algorithms can be used to obtain efficient algorithms in other distributed message-passing models such as the $k$-machine model \cite{soda15}, which is a distributed model for large-scale data
computation. We also believe that the ideas of this paper can be extended to obtain similarly fast and efficient fully-distributed algorithms for other random graph models such as the $G(n,M)$ model and random
regular graphs \cite{Bollobas-book}.

Several open questions arise from our work. First, is it possible to show non-trivial lower bounds for the HC problem in random graphs? In particular, we conjecture that our upper bounds are essentially tight (up to
polylogarithmic factors). Second, can we find a sublinear time, i.e., an algorithm running in $o(n)$ rounds for $p = \frac{c\ln n}{n}$, i.e., at the threshold; or show that this is not possible. Finally, nothing
non-trivial is known regarding upper bounds for general graphs.

\bibliographystyle{plain}
\bibliography{hc}

\end{document}